\documentclass{article}
\usepackage[utf8]{inputenc}
\usepackage[T1]{fontenc}
\usepackage[margin=1in]{geometry}
\usepackage{amsmath,amsfonts, enumerate}
\usepackage{amsthm}
\usepackage{amssymb}
\usepackage{mathtools}
\usepackage{xcolor}
\usepackage{dsfont}
\usepackage{changepage}
\usepackage{algorithm}
\usepackage{algorithmic}
\usepackage{tabularx, cite, url, hyperref}

\newcommand{\mcp}{\mathcal{P}}
\newcommand{\mcq}{\mathcal{Q}}
\newcommand{\R}{\mathbb{R}}
\allowdisplaybreaks

\newcommand*{\email}[1]{\texttt{#1}}

\newcommand{\knote}[1]{{\color{red}[{\tiny Karthik: \textbf{#1}}]\marginpar{\color{red}*}}}
\newcommand{\wnote}[1]{{\color{blue}[{\tiny Weihang: \textbf{#1}}]\marginpar{\color{blue}*}}}

\newtheorem{theorem}{Theorem}[section]
\newtheorem{lemma}{Lemma}[section]

\newtheorem{proposition}{Proposition}[section]
\newtheorem{claim}{Claim}[section]
\newtheorem{remark}{Remark}[section]

\title{Approximating submodular $k$-partition via principal partition sequence\footnote{University of Illinois, Urbana-Champaign. Email: \email{\{karthe, weihang3\}@illinois.edu}. Supported in part by NSF grants CCF-1814613 and  CCF-1907937.}}\author{Karthekeyan Chandrasekaran \and Weihang Wang}
\date{}

\begin{document}

\maketitle
\begin{abstract}
In submodular $k$-partition, the input is a submodular function $f:2^V\rightarrow \R_{\ge 0}$ (given by an evaluation oracle) along with a positive integer $k$ and the goal is to find a partition of the ground set $V$ into $k$ non-empty parts $V_1, V_2, \ldots, V_k$ in order to minimize $\sum_{i=1}^k f(V_i)$.  
Narayanan, Roy, and Patkar \cite{NRP96} designed an algorithm for submodular $k$-partition based on the principal partition sequence and showed that the approximation factor of their algorithm is $2$ for the special case of graph cut functions (which was subsequently rediscovered by Ravi and Sinha \cite{RS07}). 
In this work, we study the approximation factor of their algorithm for three subfamilies of submodular functions---namely monotone, symmetric, and posimodular and show the following results:  
\begin{enumerate}
\item The approximation factor of their algorithm for monotone submodular $k$-partition is 
$4/3$. 
This result improves on the $2$-factor that was known to be achievable for monotone submodular $k$-partition via other algorithms. 
Moreover, our upper bound of $4/3$ matches the recently shown lower bound 
under 
polynomial number of function evaluation queries \cite{San21}. Our upper bound of $4/3$ is also the first improvement beyond $2$  for a certain graph partitioning problem that is a special case of monotone submodular $k$-partition. 

\item The approximation factor of their algorithm for symmetric submodular $k$-partition is $2$. This result generalizes their approximation factor analysis beyond graph cut functions. 

\item The approximation factor of their algorithm for posimodular submodular $k$-partition is $2$. 
\end{enumerate}
We also construct an example to show that the approximation factor of their algorithm for arbitrary submodular functions is $\Omega(n/k)$.
\end{abstract}

\section{Introduction}\label{sec:intro}
A set function $f:2^V\rightarrow \R$ is submodular if $f(A)+f(B)\ge f(A\cap B) + f(A\cup B)$ for every $A, B\subseteq V$. 
An evaluation oracle for a function $f:2^V\to\R$ takes a subset $A\subseteq V$ as input and returns $f(A)$. 
We consider the \emph{submodular $k$-partition} problem defined as follows: The input consists of a non-negative submodular function $f:2^V\to\R_{\geq 0}$ on a finite ground set $V$ via an evaluation oracle and an integer $k\geq 2$. The goal is to find a partition $V_1, V_2, \ldots,V_k$ of $V$ into $k$ non-empty parts in order to minimize $\sum_{i=1}^k f(V_i)$. Namely, the goal is to compute
\[\min\left\{\sum_{i\in[k]}f(V_i):\ V_1, V_2, \ldots, V_k\text{ is a partition of }V,\ V_i\neq\emptyset\;\forall i\in[k]\right\}.\]
Throughout, we will assume that the input submodular function is non-negative and denote the size of the ground set $V$ by $n$. 
If $k=2$, then the problem reduces to the classic submodular minimization problem. 
We emphasize that our focus is on submodular $k$-partitioning when $k$ is part of input 
(see \cite{CC22-j} for a discussion of the problem for fixed constant $k$).
Submodular $k$-partition formulates several interesting partitioning problems and we will discuss some of the special cases below. 
For arbitrary submodular functions, the problem is NP-hard \cite{GH94}, does not admit a $(2-\epsilon)$-approximation assuming polynomial number of function evaluation queries \cite{San21}, does not admit a $n^{1/(\log\log n)^c}$-approximation for every constant $c$ assuming the Exponential Time Hypothesis \cite{CL20}, and 
the best approximation factor that is known is $O(k)$ \cite{ZNI05, OFN12}. 

In this work, we will be interested in the submodular $k$-partition problem for subfamilies of submodular functions---namely monotone, symmetric, and posimodular submodular functions. 
A set function $f:2^V\rightarrow \R$ is 
\begin{enumerate}
    \item monotone if $f(B)\ge f(A)$ for every $A\subseteq B\subseteq V$, 
    \item symmetric if $f(A) = f(V-A)$ for every $A\subseteq V$, and 
    \item posimodular if $f(A)+f(B)\ge f(A- B) + f(B-A)$ for every $A, B\subseteq V$. 
\end{enumerate}
If the input submodular function is monotone/symmetric/posimodular, then we call the associated submodular $k$-partition problem as monotone/symmetric/posimodular submodular $k$-partition. 
 We note that monotone submodular functions and symmetric submodular functions are also posimodular\footnote{In fact, monotone functions are posimodular since $f(A)\ge f(A-B)$ and $f(B)\ge f(B-A)$ for every $A, B\subseteq V$. Symmetric submodular functions are posimodular: for every $A, B\subseteq V$, we have that $f(A) + f(B) = f(V-A) + f(B)\ge f((V-A)\cup B) + f((V-A)\cap B)=f(V-(A-B))+f(B-A) =f(A-B)+f(B-A)$.}.  Hence, posimodular submodular $k$-partition problem generalizes both monotone submodular $k$-partition and symmetric submodular $k$-partition problems.  
 We now discuss the approximation status of symmetric/monotone/posimodular submodular $k$-partition and some of their well-known special cases  
(see Table \ref{tab: approximation factors} for a summary of approximation factors of symmetric/monotone/posimodular submodular $k$-partition achieved by different approaches).


\vspace{2mm}
\noindent \textbf{Monotone submodular $k$-partition.} 
 Special cases of monotone submodular $k$-partition problem include matroid $k$-partition and coverage $k$-partition---the submodular functions of interest here are matroid rank functions and coverage functions respectively. Matroid $k$-partition captures several interesting problems: e.g., (1) partition the columns of a given matrix into $k$ non-empty parts to minimize the sum of the dimension of the subspace spanned by the parts and (2) partition the edges of a given graph into $k$ non-empty parts to \emph{maximize} the sum of the number of connected components formed by the parts. 
 Coverage $k$-partition also captures several interesting problems: e.g., (3) partition the vertices of a given graph into $k$ non-empty parts $V_1, V_2, \ldots, V_k$ in order to minimize $\sum_{i=1}^k f(V_i)$, where $f(P)$ is the number of edges incident to the vertex subset $P\subseteq V$. 
 To gain a better understanding of the difficulty in solving/approximating monotone submodular $k$-partition, we encourage the reader to briefly think about the concrete special case of matrix column partitioning problem 
 (i.e., problem (1) described above which is seemingly a linear algebra problem) before reading further. 

 Monotone submodular $k$-partition is NP-hard \cite{San21}. 
 Moreover, it admits a simple (and fast) $(2-1/k)$-approximation algorithm that will be denoted henceforth as the \emph{cheapest singleton partitioning algorithm}: return the partition $V_1:=\{v_1\}, V_2:=\{v_2\}, \ldots, V_{k-1}:=\{v_{k-1}\}, V_k:=V-\{v_1, \ldots, v_{k-1}\})$, where the $n$ elements of the ground set are ordered as $v_1, \ldots, v_{n}$ such that $f(\{v_1\})\le f(\{v_2\}) \le \ldots \leq f(\{v_n\})$. Santiago \cite{San21} showed that this is a $2$-approximation. 
 His analysis can be extended to show that it is in fact a $(2-1/k)$-approximation\footnote{The cheapest singleton partitioning algorithm returns a solution whose cost is $f(V-V_k)+\sum_{i=1}^{k-1} f(\{v_i\})\le f(V) + \sum_{i=1}^{k-1} f(\{v_i\})\le f(V) + (1-1/k)\sum_{i=1}^{k} f(\{v_i\})\le (2-1/k)\max\{f(V), \sum_{i=1}^k f(\{v_i\})\}$ while the cost of an optimum $k$-partition is at least $\max\{f(V), \sum_{i=1}^k f(\{v_i\})\}$. The lower bound on the cost of the optimum $k$-partition $V_1^*, \ldots, V_k^*$ is because $\sum_{i=1}^k f(V_i^*)\ge f(V)$ by non-negativity and submodularity and moreover, if the optimum partition is indexed such that $\min\{j\in [n]: v_j\in V_i^*\}\le \min\{j\in [n]: v_j\in V_{i+1}^*\}$ for all $i\in [k-1]$, then $f(V_i^*)\ge f(\{v_i\})$ by monotonicity and hence, $\sum_{i=1}^k f(V_i^*)\ge \sum_{i=1}^{k} f(\{v_i\})$.} and this is the best possible approximation factor for this algorithm\footnote{The best possible approximation factor for the cheapest singleton partitioning algorithm is $2-1/k$ as seen from this example: Let $f$ be the rank function of a partition matroid on a $k$-partition $\{S_1,\ldots, S_k\}$ of the ground set $S$ where $|S_i|\geq 2$ for all $i\in[k]$. Then, the algorithm may return $\{\{s_1\}, \{s_2\},\ldots, \{s_{k-1}\}, S-\{s_1, \ldots, s_{k-1}\}\}$, where $s_i\in S_i$ for all $i\in[k-1]$ and this $k$-partition has objective value $2k-1$, whereas the partition $\{S_1,\ldots, S_k\}$ has objective value $k$.}. 
 Alternatively, the greedy splitting algorithm presented in \cite{ZNI05} achieves a $(2-2/k)$-approximation. 
On the inapproximability front, Santiago \cite{San21} showed that there does not exist an algorithm that makes polynomial number of function evaluation queries to obtain a $(4/3-\epsilon)$-approximation for every constant $\epsilon>0$.

\vspace{2mm}
\noindent \textbf{Symmetric submodular $k$-partition.} 
Well-known special cases of symmetric submodular $k$-partition problem are graph $k$-cut and hypergraph $k$-partition---the 
submodular functions of interest here are the cut functions of an explicitly given graph and hypergraph respectively.
Graph $k$-cut is NP-complete \cite{GH94} and does not have a polynomial-time $(2-\epsilon)$-approximation for every constant $\epsilon>0$ under the Small Set Expansion Hypothesis \cite{Ma18}.
There are several known approaches to achieve a $2$-approximation for graph $k$-cut---(i) 
greedy splitting approach \cite{SV95}, (ii) Gomory-Hu tree based approach \cite{Va03}, (iii) extreme sets based approach \cite{NK07}, (iv) principal partition sequence based approach \cite{NRP96, Bar00, RS07}, and (v) covering-LP based approach \cite{NR01, CQX20, Q19}. Greedy splitting, Gomory-Hu tree, and extreme sets based approaches lead to a $(2-2/k)$-approximation while the principal partition sequence and the covering-LP based approaches lead to a $(2-2/n)$-approximation for graph $k$-cut.  
The principal partition sequence and the covering-LP based approaches for graph $k$-cut have also been shown to be related to each other \cite{CQX20}. 
The principal partition sequence based approach is the main algorithm of interest to our work and we will discuss it in detail in Section \ref{sec:pps}.

For the more general problem of symmetric submodular $k$-partition, two of the  approaches discussed in the previous paragraph for graph $k$-cut have been generalized to obtain $2$-approximations---the greedy splitting approach \cite{ZNI05} and the Gomory-Hu tree approach lead to a $(2-2/k)$-approximation. 
Analyzing the approximation factor of the principal partition sequence based approach for symmetric submodular $k$-partition was one of the driving motivations of our work. 
On the inapproximability front, Santiago \cite{San21} showed that there does not exist an algorithm that makes polynomial number of function evaluation queries to obtain a $(2-\epsilon)$-approximation for every constant $\epsilon >0$. 

\vspace{2mm}
\noindent \textbf{Posimodular submodular $k$-partition.} The only natural family of posimodular submodular functions that we are familiar with are symmetric submodular functions and monotone submodular functions as well as their positive linear combinations.
 As mentioned before, posimodular submodular $k$-partition is a unified generalization of symmetric submodular $k$-partition and  monotone submodular $k$-partition.  
 To the best of authors' knowledge, posimodular submodular $k$-partition has not been studied in the literature before and there are no specialized algorithms or approximation factor analysis of existing algorithms for posimodular submodular $k$-partition.  
 A slight modification to the analysis of the greedy splitting algorithm presented in \cite{ZNI05} shows that their algorithm achieves a $(3-2/k)$-approximation for posimodular submodular $k$-partition---we refrain from presenting this analysis in the interests of brevity. On the inapproximability front, since symmetric submodular functions are also posimodular submodular, the lower bound for symmetric submodular $k$-partition also holds for posimodular submodular $k$-partition, i.e., there does not exist an algorithm for posimodular submodular $k$-partition that makes polynomial number of function evaluation queries to obtain a $(2-\epsilon)$-approximation for every constant $\epsilon>0$.   
 


\subsection{Our Results}\label{sec:results}
In this work, we investigate Narayanan, Roy, and Patkar's \cite{NRP96} principal partition sequence based algorithm for submodular $k$-partition. They showed that their algorithm achieves a $2$-approximation for graph $k$-cut (which was subsequently rediscovered by Ravi and Sinha \cite{RS07}).
We show the following results: 
 \begin{enumerate}
 
 \item Their algorithm achieves a 
 $4/3$-approximation
 for monotone submodular $k$-partition. This result improves on the $2$-factor 
 that is known to be achievable via two different algorithms: the cheapest singleton partitioning algorithm and 
 the greedy splitting algorithm. Moreover, our upper bound of $4/3$ matches the lower bound shown by Santiago \cite{San21}. We will discuss the significance of our upper bound result shortly. 

 \item Their algorithm achieves a $2$-approximation 
 for symmetric submodular $k$-partition. This factor matches the $2$-factor that is known to be achievable via two other algorithms: the greedy splitting algorithm and the Gomory-Hu tree based algorithm, and also matches the lower bound \cite{Ma18, San21}. Our contribution here is generalizing the analysis of \cite{NRP96, RS07} to beyond graph cut functions. 
 
 \item Their algorithm achieves a $2$-approximation 
 for posimodular submodular $k$-partition. This result improves on the $3$-factor that is known to be achievable via the greedy splitting algorithm and matches the lower bound of $2$ shown by Santiago \cite{San21}.

 \end{enumerate}
 See Table \ref{tab: approximation factors} for a comparison. 
 Graph $k$-cut is the well-studied special case of symmetric submodular $k$-partition/posimodular submodular $k$-partition, so we include that as the last column in the table for comparison.  Approximation factors in the row corresponding to principal partition sequence are the main results of our work. In the last row of the table, we include the known lower bounds on the approximation factor for comparison. The lower bound for graph cut function is assuming the Small Set Expansion Hypothesis \cite{Ma18} while the rest of the lower bounds are assuming polynomial number of function evaluation queries. Dashes in the table indicate that either the approach does not extend or there has been no analysis of the approximation factor of the approach for the subfamily.
 
 We complement our upper bounds on the approximation factor of their algorithm with matching lower bound constructions for each subfamily of submodular functions. 
 Our results show that the principal partition sequence based algorithm achieves the best possible approximation factor for broad subfamilies of submodular functions, thus illustrating the power and applicability of this algorithm. 
 On the other hand, we show that the approximation factor of their algorithm for arbitrary submodular functions is $\Omega(n/k)$ via a lower bound construction. This construction shows that their principal partition sequence based algorithm cannot directly improve the approximation factor for submodular $k$-partition beyond the current best $O(k)$. 
We briefly discuss the significance of our $4/3$-approximation result for monotone submodular $k$-partition. Firstly, prior to our results, there were no known families of submodular functions for which the submodular $k$-partition problem could be approximated to a factor better than $2$. Our result for monotone submodular functions breaks this $2$-factor barrier for a broad family of submodular functions. Secondly, our result for monotone submodular $k$-partition leads to a new approximation result even for a graph partitioning problem that we describe now. 
For a graph $G=(V, E)$ with edge weights $w:E\rightarrow \R_+$, consider functions $d, f: 2^V\rightarrow \R_+$ defined by $d(S):=w(\delta(S))$ and $f(S):=w(E[S])+w(\delta(S))$ for every $S\subseteq V$, where $\delta(S)$ denotes the set of edges with exactly one end-vertex in $S$, $E[S]$ denotes the set of edges with both end-vertices in $S$, and $w(F):=\sum_{e\in F}w(e)$ for every $F\subseteq E$. The function $d$ is the cut function of the graph and is symmetric submodular. The function $f$ is the coverage function of the graph and is monotone submodular. Submodular $k$-partition for the function $d$ is known as graph $k$-cut and it is known that graph $k$-cut does not admit a $(2-\epsilon)$-approximation under the Small Set Expansion Hypothesis \cite{Ma18}. In contrast, our results show that coverage $k$-partition in graphs---i.e., submodular $k$-partition for the function $f$---admits a $4/3$-approximation. We note that coverage $k$-partition in graphs is NP-hard \cite{San21} and its approximability is an intriguing open question.

\begin{table}[h]
    \centering
    \begin{tabular}{|c|c|c|c|c|}
    \hline
    & Monotone & Symmetric & Posimodular & Graph  \\ 
    & Submodular & Submodular & Submodular & Cut Function \\ 
    & Function & Function & Function & \\ 
    \hline
    Greedy splitting & $2-2/k$ \cite{ZNI05} & $2-2/k$ \cite{ZNI05} &$3-2/k$ \cite{ZNI05}* & $2-2/k$ \cite{SV95} \\
    \hline
    Extreme Sets&---&---&---& $2-2/k$ \cite{NK07}\\
    \hline
    Gomory-Hu tree & --- & $2-2/k$ [Folklore] & --- & $2-2/k$ \cite{Va03} \\
    \hline
    Covering-LP &---&---&---&$2-2/k$ \cite{NR01, CQX20}\\
    \hline
    Cheapest Singleton &$2-1/k$ \cite{San21}&---&---&---\\
     Partitioning & & & & \\
     \hline
    Principal Partition & $4/3-4/(9n+3)$ & $2-2/n$ & $2-2/(n+1)$ & $2-2/n$ \cite{NRP96, RS07} \\
    Sequence & (Theorem \ref{Thm: 4/3-approx mono sm k-cut algo}) & (Theorem \ref{Thm: 2-approx sym sm k-cut algo}) & (Theorem \ref{Thm: 2-approx pm sm k-cut algo}) & \\
    \hline
    \textbf{Lower Bound} & $4/3-\epsilon$ \cite{San21} & $2-\epsilon$ \cite{San21}& $2-\epsilon$ \cite{San21} & \textbf{$2-\epsilon$} \cite{Ma18}\\
    \hline
    \end{tabular}
    \caption{Approximation factors of symmetric/monotone/posimodular submodular $k$-partition using different approaches. Result in the first row marked with an asterisk follows by slight modifications to the known analysis of the approximation factor for symmetric submodular functions given in \cite{ZNI05}.}
    \label{tab: approximation factors}
\end{table}

\vspace{2mm}
\noindent \textbf{Organization.} We discuss the principal partition sequence based algorithm in Section \ref{sec:pps}. We analyze the approximation factor of the algorithm with matching lower bound constructions for each of the three subfamilies of submodular functions in Section \ref{sec:approx-analysis}. We exhibit an instance of submodular $k$-partition where the algorithm achieves an approximation factor of $\Omega(n/k)$ in Section \ref{subsec: PPS n/k-apx}.

 \subsection{Related work} 
 The principal partition sequence based algorithm for submodular $k$-partition was introduced by Narayanan, Roy, and Patkar \cite{NRP96}. We will formally define the principal partition sequence of a submodular function and describe their algorithm in Section \ref{sec:pps}. They analyzed the approximation factor of their algorithm for two variants of $k$-partitioning problems in hypergraphs. 
 These two variants are not special cases of symmetric/monotone/posimodular submodular $k$-partition and are not of direct interest to our work. However, we describe these variants 
 to highlight 
 the versatility 
 of the principal partition sequence based approach and also to shed light on the results of Narayanan, Roy, and Patkar's work which do not seem to be well-known in the literature. 
Given a hypergraph $H=(V,E)$, a hyperedge cost function $c:E\to\R_+$, and an integer $k$, the goal is to find a partition
$\mcp:=\{V_1,V_2, \ldots,V_k\}$ of $V$ 
into $k$ non-empty parts 
that minimizes an objective of interest: 
\begin{enumerate}
    \item If the objective is the sum of cost of hyperedges that intersect at least two parts of $\mcp$, then the problem is known as \emph{hypergraph $k$-cut}. 
    \item If the objective is the sum of \emph{cost of hyperedges relative to the partition $\mcp$}, where the cost of a hyperedge $e$ relative to $\mcp$ is $c(e)(\ell-1)$ with $\ell$ being the number of parts of $\mcp$ intersected by $e$, then the problem is known as \emph{normalized coverage $k$-partition}\footnote{We introduce this nomenclature because the problem is equivalent to finding a partition $V_1, V_2, \ldots, V_k$ of the ground set $V$ in order to minimize $\sum_{i=1}^k f(V_i) - f(V)$, where $f:2^V\rightarrow \R_+$ is an explicitly given coverage function (every coverage function can be uniquely represented using a hypergraph \cite{CH12}). We consider the subtraction of $f(V)$ as normalizing the objective since it is a 
    trivial lower bound on the sum of the function values of the parts: $\sum_{i=1}^k f(V_i)\ge f(V)$ holds for every $k$-partition $V_1, V_2, \ldots, V_k$ since $f$ is a coverage function.}.  
\end{enumerate}
Narayanan, Roy, and Patkar \cite{NRP96} 
showed that their principal partition sequence based algorithm achieves 
a $r(1-1/n)$-approximation for hypergraph $k$-cut, where $r$ is the size of the largest hyperedge and $n$ is the number of vertices in the input hypergraph, and achieves a $(2-2/n)$-approximation for normalized coverage $k$-partition.
A consequence (of both of their results) is that the principal partition sequence based algorithm achieves a $(2-2/n)$-approximation for graph $k$-cut. 
Their principal partition sequence based algorithm for graph $k$-cut is equivalent to the Lagrangean relaxation approach suggested by Barahona \cite{Bar00}. The approximation factor of the principal partition sequence based algorithm for graph $k$-cut being at most $2$ was 
rediscovered by Ravi and Sinha \cite{RS07} and for hypergraph $k$-cut being at most $r$ was rediscovered by Ba\"{i}ou and Barahona \cite{BB23}. 

We mention that a slight modification to the analysis of the greedy splitting algorithm presented in \cite{ZNI05} shows that the greedy splitting  algorithm achieves a $(2-2/k)$-approximation for normalized coverage $k$-partition and hypergraph $k$-partition. We note that hypergraph $k$-partition is a special case of symmetric submodular $k$-partition and is different from hypergraph $k$-cut (for definition of hypergraph $k$-partition, see discussion of symmetric submodular $k$-partition at the beginning of the introduction). 
On the inapproximability front, it is known that hypergraph $k$-cut does not admit an approximation factor of $n^{1/(\log\log n)^c}$, where $c$ is a constant, assuming the Exponential Time Hypothesis \cite{CL20}. The best inapproximability result for the more general submodular $k$-partition problem (mentioned in the introduction) follows from this inapproximability for hypergraph $k$-cut.

\section{Principal partition sequence based algorithm}
\label{sec:pps}
In this section, we recall the principal partition sequence based algorithm for submodular $k$-partition designed by Narayanan, Roy, and Patkar \cite{NRP96}. We begin with some notation. 
Throughout this work, a \emph{partition} of a set $S$ is defined to be a collection of \emph{nonempty} pairwise disjoint subsets of $S$ whose union is $S$, and a \emph{$k$-partition} of a set $S$ is defined to be a partition of $S$ with exactly $k$ parts. For two distinct partitions $\mcp$ and $\mcq$ of a set $S$, if every part of $\mcq$ is completely contained in some part of $\mcp$ then we say that $\mcq$ \emph{refines} $\mcp$ (equivalently, $\mcp$ is a coarsening of $\mcq$). For two distinct partitions $\mcp$ and $\mcq$ of a set $S$, we will say that $\mcq$ is obtained from $\mcp$ by refining only one part of $\mcp$ if there exists a part $P\in \mcp$ such that $P\notin \mcq$ and every part $Q\in \mcq$ satisfies either $Q\subsetneq P$ or $Q\in \mcp$ (i.e., either $Q$ is a proper subset of the part $P$ or $Q$ is a part of the partition $\mcp$); we will denote such a part $P\in \mcp$  as the part refined by $\mcq$. 

Let $f:2^V\to\R$ be a set function on ground set $V$. For a collection $\mcp$ of subsets of $V$, we write $f(\mcp):=\sum_{P\in \mcp}f(P)$. We will say that a partition $\mcp =\{P_1, \ldots, P_k\}$ is an optimal $k$-partition if $f(\mcp)\le f(\mcq)$ for every $k$-partition $\mcq$ of $V$. 
We define the function $g_{f,\mcp}:\R_{\geq 0}\to\R$ for a partition $\mcp$ of the ground set $V$ and the function $g_f:\R_{\geq 0}\to\R$ as follows:
\begin{align*}
    g_{f,\mcp}(b)&:=f(\mcp)-b|\mcp| \text{ and}\\
    g_f(b)&:=\min\{g_{f,\mcp}(b):\mcp\text{ is a partition of }V\}.
\end{align*}
We drop the subscript $f$ and instead write $g_{\mcp}$ and $g$ respectively, if the function $f$ is clear from context. 
By definition, the function $g_f$ is piece-wise linear. It can be shown that $g_f$ has at most $|V|-1$ breakpoints. 
The next theorem shows that if the function $f:2^V\rightarrow \R$ is submodular, then 
there exists a sequence of partitions achieving the $g_f$ function values at the breakpoints that have a nice structure; moreover, the breakpoints and such a sequence of partitions can be computed in polynomial time given access to the evaluation oracle of the submodular function $f$. 
We emphasize that the theorem holds for arbitrary submodular functions (which may not be non-negative valued). 

\begin{theorem}[\hspace{-1sp}\cite{Na91,NRP96}] \label{Thm: PPS}
Let $f:2^V\to\R$ be a submodular function on a ground set $V$. 
Then, there exists a sequence $\mcp_1, \mcp_2, \ldots, \mcp_{r}$ of partitions of $V$ and values $b_1, b_2, \ldots, b_{r-1}$ such that
\begin{enumerate}
    \item $\mcp_1=\{V\}$ and $\mcp_r=\{\{v\}:v\in V\}$,
    \item For each $j\in[r-1]$, the partition $\mcp_{j+1}$ is obtained from $\mcp_j$ by refining only one part of $\mcp_j$, 
    \item $b_1< b_2< \ldots < b_{r-1}$, 
    \item $g(b_j)=g_{\mcp_j}(b_j)=g_{\mcp_{j+1}}(b_j)$ for each $j\in[r-1]$ and  
    \item 
    $g(b)=g_{\mcp_1}(b)$ for all $b\in (-\infty, b_1]$, \\
    $g(b) = g_{\mcp_{j+1}}(b)$ for all $b\in [b_j, b_{j+1}]$ for each $j\in [r-2]$, and \\
    $g(b)=g_{\mcp_r}(b)$ for all $b\in [b_{r-1}, \infty)$. 
\end{enumerate}
Moreover, such a sequence $\mcp_1, \mcp_2, \ldots, \mcp_{r}$ of partitions of $V$ and values $b_1, b_2, \ldots, b_{r-1}$ can be computed in polynomial time given access to the evaluation oracle of the submodular function $f$. 

\end{theorem}

For a submodular function $f:2^V\rightarrow \R$, we will denote a sequence of partitions $\mcp_1, \mcp_2, \ldots,\mcp_r$ and the sequence of values $b_1,b_2, \ldots, b_{r-1}$ satisfying the conditions given in Theorem \ref{Thm: PPS} as a \emph{principal partition sequence} and the \emph{critical value sequence} of $f$, respectively. 
We note that this definition differs from those in \cite{Na91,NRP96} owing to the reversed indexing order and the imposition of condition 2---we note that the proofs given in those papers also show that condition 2 holds (also see \cite{RS07}). 
The principal partition sequence of submodular functions is known in the literature as \emph{principal lattice of partitions} of submodular functions since there exists a lattice structure associated with the sequence of partitions. We choose to call it as principal partition sequence in this work since the sequence suffices for our purpose. 
For more on principal lattice of partitions of submodular functions and their computation, we refer the reader to \cite{Cun85, Na91, NRP96, Narayanan-book, DNP03, PN03, Bar00, Kol10, NKI10}. 

We now discuss the principal partition sequence based algorithm for submodular $k$-partition that was proposed by Narayanan, Roy, and Patkar \cite{NRP96}. 
This algorithm computes a principal partition sequence 
satisfying all conditions in Theorem \ref{Thm: PPS}. 
If the sequence contains a partition that has exactly $k$ parts, then the algorithm returns this $k$-partition. Otherwise, the algorithm returns a $k$-partition obtained by refining the partition in the sequence that has the largest number of parts that is less than $k$. 
The refinement is based on the partition in the sequence that has the fewest number of parts that is more than $k$. 
The formal description of the refinement is given in Algorithm 1. Since the sequence $\mcp_1,\mcp_2,\ldots,\mcp_r$ satisfying the conditions of Theorem \ref{Thm: PPS} can be computed in polynomial time, Algorithm 1  
can indeed be implemented to run in polynomial time. By design, the algorithm returns a $k$-partition. The remainder of this work will focus on analyzing the approximation factor of the algorithm. 
\begin{algorithm}[ht] \label{Algo: 2-approx sym sm k-partition}
\begin{algorithmic}
\STATE \textbf{Input: }A submodular function $f:2^V\to\R$ given by evaluation oracles and an integer $k\geq 2$.
\STATE \textbf{Output: }A $k$-partition $\mcp$ of $V$.
\STATE Use Theorem \ref{Thm: PPS} to compute a principal partition sequence $\mcp_1,\mcp_2,\ldots,\mcp_{r}$ of the submodular function $f$ satisfying all conditions stated in that theorem.
\IF{ $|\mcp_j|=k$ for some $j\in[r]$}
\STATE Return $\mcp:=\mcp_j$.
\ENDIF
\STATE Let $i\in\{2,3,\ldots,r\}$ so that $|\mcp_{i-1}|<k<|\mcp_i|$.
\STATE Let $S\in \mcp_{i-1}$ be the part refined by $\mcp_i$ and $\mcp'$ be the parts of $\mcp_i$ contained in $S$.
\STATE Let $\mcp'=\{B_1,\ldots,B_{|\mcp'|}\}$ such that $f(B_1)\leq f(B_2)\leq\ldots\leq f(B_{|\mcp'|})$.
\STATE Return $\mcp:=(\mcp_{i-1}\backslash\{S\})\cup\left\{B_1,B_2,\ldots,B_{k-|\mcp_{i-1}|}\right\}\cup\left\{\bigcup_{j=k-|\mcp_{i-1}|+1}^{|\mcp'|}B_j\right\}$.
\end{algorithmic}
\caption{Principal partition sequence based algorithm for submodular $k$-partition}
\end{algorithm}

To construct examples that exhibit tight lower bound on the approximation factor of Algorithm 1, we will need the following proposition that identifies a special case under which the principal partition is unique and consists only of two partitions---namely, the partition into singletons and the partition that consists of only one part.  
\begin{proposition}\label{prop:unique-principal-partition}
Let $f: 2^V\rightarrow \R_{\ge 0}$ be a non-negative submodular function. Suppose that for every partition $\mcp \neq \mcq, \{V\}$ where $\mcq:=\{\{v\}:v\in V\}$, the function $f$ satisfies 
\[
\frac{f(\mcp)-f(V)}{|\mcp|-1}>\frac{f(\mcq)-f(V)}{|V|-1}.
\]
Then, the principal partition sequence of $f$ is $\{V\}, \mcq$. 
\end{proposition}
\begin{proof}
Let $n:=|V|$. It suffices to show that for every partition $\mcp$ of $V$ such that $\mcp\neq\{V\},\mcq$, we have that
\begin{align}
    f(\mcp)-b|\mcp|>\min\{f(V)-b, f(\mcq)-b\cdot n\}\quad\forall b\in\R \label{eqn: 4/3 tight example - 1}
\end{align}
This suffices since it ensures that no partitions other than $\{V\}$ and $\mcq$ satisfy conditions 4 and 5 in Theorem \ref{Thm: PPS}.
By the hypothesis, 
we have 
\[\frac{f(\mcp)-f(V)}{|\mcp|-1}> \frac{f(\mcq)-f(V)}{n-1},\]
which is equivalent to 
\begin{align}
    f(\mcp)-b'|\mcp|>f(V)-b',\label{eqn: 4/3 tight example - 2}
\end{align}
where $b'=(f(\mcq)-f(V))/(n-1)$ for every partition $\mcp\neq\{V\},\mcq$. Now, suppose inequality \eqref{eqn: 4/3 tight example - 1} fails for some partition $\mcp$ and some $b\in \R$, then
\begin{align}
    &f(\mcp)-b|\mcp|\leq f(V)-b \text{ and } \label{eqn: 4/3 tight example - 3}
    \\& f(\mcp)-b|\mcp|\leq f(\mcq)-bn. \label{eqn: 4/3 tight example - 4}
\end{align}
Consequently, we have that
\begin{align*}
    (b-b')|\mcp| 
    &= (f(\mcp)-b'|\mcp|)-(f(\mcp)-b|\mcp|) \\
    &> (f(V)-b') -(f(\mcp)-b|\mcp|) \quad \quad \text{(using inequality \eqref{eqn: 4/3 tight example - 2})}\\
    &\ge (f(V)-b') - (f(V)-b) \quad \quad \text{(using inequality \eqref{eqn: 4/3 tight example - 3})}\\
    &= b-b'. 
\end{align*}
This implies $b-b'>0$. Moreover, we also have that 
\begin{align*}
    (b-b')|\mcp| 
    &=(f(\mcp)-b'|\mcp|)-(f(\mcp)-b|\mcp|) \\
    &>(f(V)-b') -(f(\mcp)-b|\mcp|) \quad \quad \text{(using inequality \eqref{eqn: 4/3 tight example - 2})}\\
    &\ge (f(V)-b') - (f(\mcq)-bn) \quad \quad \text{(using inequality \eqref{eqn: 4/3 tight example - 4})}\\
    & = (f(\mcq)-b'n)- (f(\mcq)-bn)\quad\quad(\text{by definition of }b')\\
    & = (b-b')n.
\end{align*}
Since $b-b'> 0$, it follows that $|\mcp|>n$, which is a contradiction to the fact that $\mcp$ is a partition of $V$. Therefore, inequality \eqref{eqn: 4/3 tight example - 1} holds by contradiction.
\end{proof}
\section{Approximation factor analysis}\label{sec:approx-analysis}
In this section, we analyze the approximation factor of Algorithm 1 for various subfamilies of submodular functions. We state and prove certain lemmas that will be useful for all submodular functions (Lemmas \ref{claim:exact-k} and \ref{lemma: f long ineq}). 
We analyze the approximation factor of Algorithm 1 for monotone submodular functions in Section \ref{sec:mono}, for symmetric submodular functions in Section \ref{sec:sym}, and for posimodular submodular functions in Section \ref{sec:pm}. We conclude each subsection with a remark on the tightness of the approximation factor for the algorithm. 

Our first lemma identifies a special case in which Algorithm 1 returns an optimum $k$-partition. This special case was also identified by Narayanan, Roy, and Patkar. We note that the following lemma holds for arbitrary submodular functions (which may not be non-negative). 

\begin{lemma}[\hspace{-1sp}\cite{NRP96}]\label{claim:exact-k}
Let $k\geq 2$ be an integer, $f:2^V\to\R$ be a submodular function on a ground set $V$, and $\mcp_1, \mcp_2, \ldots, \mcp_r$ be a principal partition sequence of the submodular function $f$ satisfying the conditions of Theorem \ref{Thm: PPS}. 
If there exists $j\in [r]$ such that $|\mcp_j|=k$, then $\mcp_j$ is an optimal $k$-partition.
\end{lemma}
\begin{proof}
Let $\mcp^\ast$ be a $k$-partition of $V$ that minimizes $f(\mcp^\ast)$ and let $b_j$ be the value where $g(b_j)=g_{\mcp_j}(b_j)$. Then, 
\[f(\mcp^\ast)-b_j\cdot k=g_{\mcp^{\ast}}(b_j)\geq g(b_j) = g_{\mcp_j}(b_j)=f(\mcp_j)-b_j|\mcp_j|=f(\mcp_j)-b_j\cdot k,\]
and hence, $f(\mcp^\ast)\geq f(\mcp_j)$. 
Therefore, $\mcp_j$ is indeed an optimal $k$-partition.
\end{proof}

In order to address the case where there is no $j\in [r]$ such that $|\mcp_j|=k$, we need the following lemma that shows two lower bounds on the optimum value. The first lower bound in the lemma holds for arbitrary submodular functions (which may not be non-negative) while the second lower bound holds for non-negative submodular functions. 
\begin{lemma} \label{lemma: f long ineq} 
Let $k\geq 2$ be an integer, $f:2^V\to\R$ be a submodular function on a ground set $V$, $\mcp^\ast$ be a $k$-partition of $V$ that minimizes $f(\mcp^\ast)$, and $\mcp_1, \mcp_2, \ldots, \mcp_r$ be a principal partition sequence of the submodular function $f$ satisfying the conditions of Theorem \ref{Thm: PPS}. 
Suppose $|\mcp_j|\neq k$ for all $j\in[r]$.
Let $\mcp_{i-1},\mcp_i$ be the partitions such that $|\mcp_{i-1}|<k<|\mcp_i|$. Then, 
\begin{enumerate}[(i)]
    \item 
$f(\mcp^\ast) 
\ge 
\frac{|\mcp_i|-k}{|\mcp_i|-|\mcp_{i-1}|}f(\mcp_{i-1})+\frac{k-|\mcp_{i-1}|}{|\mcp_i|-|\mcp_{i-1}|}f(\mcp_i)$ and
\item $f(\mcp^\ast) \ge f(\mcp_{i-1})$ if $f$ is non-negative. 
\end{enumerate}
\end{lemma}
\begin{proof}
We prove the two lower bounds below. 
\begin{enumerate}[(i)]
    \item 
Let $b_{i-1}$ be the value such that $g_{\mcp_{i-1}}(b_{i-1})=g_{\mcp_{i}}(b_{i-1})$. Then, we have
\begin{align}
    f(\mcp_{i-1})-b_{i-1}|\mcp_{i-1}|=f(\mcp_{i})-b_{i-1}|\mcp_{i}|\implies b_{i-1}=\frac{f(\mcp_i)-f(\mcp_{i-1})}{|\mcp_i|-|\mcp_{i-1}|}.\label{Eqn: b0}
\end{align}
By condition 4 of Theorem \ref{Thm: PPS}, we also have
\begin{align}
    f(\mcp^\ast)-b_{i-1}\cdot k&=g_{\mcp^\ast}(b_{i-1})\geq g(b_{i-1})=g_{\mcp_i}(b_{i-1})= f(\mcp_i)-b_{i-1}|\mcp_i|\notag
    \\&\implies f(\mcp^\ast)\geq f(\mcp_i)+b_{i-1}(k-|\mcp_i|).\label{Eqn: f(P*)}
\end{align}
Combining \eqref{Eqn: b0} and \eqref{Eqn: f(P*)}, we get
\begin{align*}
    f(\mcp^\ast)&\geq f(\mcp_i)+\frac{f(\mcp_i)-f(\mcp_{i-1})}{|P_i|-|P_{i-1}|}(k-|\mcp_i|)
    \\&=\frac{|\mcp_i|-k}{|\mcp_i|-|\mcp_{i-1}|}f(\mcp_{i-1})+\frac{k-|\mcp_{i-1}|}{|\mcp_i|-|\mcp_{i-1}|}f(\mcp_i).
\end{align*}
\item Let $P_1^*, \ldots, P_k^*$ be the parts of $\mcp^{\ast}$ (in arbitrary order). Let $k':=|\mcp_{i-1}|$. We know that $k'<k$. Consider the $k'$-partition $\mcq$ obtained as $Q_1:=P_1^*, Q_2:=P_2^*, \ldots, Q_{k'-1}:=P_{k'-1}^*, Q_{k'}:=\cup_{j=k'}^{k}P_j^*$. Then, 
\[
f(\mcp^{\ast})
= \sum_{i=1}^{k} f(P_i^*)
\ge \left(\sum_{i=1}^{k'-1} f(P_i^*)\right)+f(\cup_{j=k'}^{k}P_j^*)
= \sum_{i=1}^{k'}f(Q_i) 
= f(\mcq). 
\] 
The inequality above is due to submodularity and non-negativity. 
The partition $\mcq$ is a $k'$-partition. By Lemma \ref{claim:exact-k}, the partition $\mcp_{i-1}$ is an optimal $k'$-partition. Hence, 
\[
f(\mcq) \ge f(\mcp_{i-1}).
\]
The above two inequalities together imply that $f(\mcp^\ast) \ge f(\mcp_{i-1})$. 
\end{enumerate}
\end{proof}

\subsection{Monotone submodular functions}\label{sec:mono}
In this section, we bound the approximation factor of Algorithm 1 for monotone submodular $k$-partitioning. The following is the main theorem of this section. 

\begin{theorem} \label{Thm: 4/3-approx mono sm k-cut algo}
The approximation factor of Algorithm 1 for non-negative monotone  submodular $k$-partitioning is $\frac{4}{3}-\frac{4}{9n+3}$, where $n$ is the size of the ground set. 
\end{theorem}

The asymptotic approximation factor of $4/3$ achieved by Algorithm 1 is the best possible for non-negative monotone submodular $k$-partition: for every constant $\epsilon>0$, there does not exist an algorithm that achieves a $(4/3-\epsilon)$-approximation using polynomial number of function evaluation queries \cite{San21}. 
We will also exhibit examples to show the tightness of the approximation factor for Algorithm 1 after proving the theorem. 
The proof of Theorem  \ref{Thm: 4/3-approx mono sm k-cut algo} follows from Lemma \ref{claim:exact-k} and Lemma \ref{lemma:monotone 4/3} shown below. 

\begin{lemma}\label{lemma:monotone 4/3}
Let $k\geq 2$ be an integer, $f:2^V\to\R_{\ge 0}$ be a non-negative monotone submodular function on a ground set $V$ of size $n$, $\mcp^\ast$ be a $k$-partition of $V$ that minimizes $f(\mcp^\ast)$, and $\mcp_1, \mcp_2, \ldots, \mcp_r$ be a principal partition sequence of the submodular function $f$ satisfying the conditions of Theorem \ref{Thm: PPS}. 
Suppose $|\mcp_j|\neq k$ for all $j\in[r]$.
Then, the partition $\mcp$ returned by Algorithm 1
satisfies 
\[f(\mcp)\leq \left(\frac{4}{3}-\frac{4}{9n+3}\right)f(\mcp^\ast).\]
\end{lemma}

\begin{proof}
Let $\mcp_{i-1},\mcp_i$ be the partitions such that $|\mcp_{i-1}|<k<|\mcp_i|$. Let $S$ and $\mcp'=\{B_1,B_2,\ldots,B_{|\mcp'|}\}$ be as in Algorithm 1. 

Firstly, since $\cup_{j=k-|\mcp_{i-1}|+1}^{|\mcp'|} B_j \subseteq S$ and $f$ is monotone, we have that 
\[
f\left(\cup_{j=k-|\mcp_{i-1}|+1}^{|\mcp'|} B_j\right) \le f(S).
\]
Secondly, by our choice of $B_1,B_2,\ldots,B_{k-|\mcp_{i-1}|}$, we know that
\begin{align*}
    \sum_{j=1}^{k-|\mcp_{i-1}|}f(B_j)\leq\frac{k-|\mcp_{i-1}|}{|\mcp'|}f(\mcp').
\end{align*}
Hence,
\begin{align*}
    f(\mcp)&=f(\mcp_{i-1})-f(S)+\sum_{j=1}^{k-|\mcp_{i-1}|}f(B_j)+f\left(\bigcup_{j=k-|\mcp_{i-1}|+1}^{|\mcp'|}B_j\right)
    \\&\leq f(\mcp_{i-1})+\sum_{j=1}^{k-|\mcp_{i-1}|}f(B_j)
    \\&\leq f(\mcp_{i-1})+\frac{k-|\mcp_{i-1}|}{|\mcp'|}f(\mcp')
    \\&=f(\mcp_{i-1})+\frac{k-|\mcp_{i-1}|}{|\mcp_i|-|\mcp_{i-1}|+1}(f(\mcp_i)+f(S)-f(\mcp_{i-1})),
\end{align*}
where the last equality follows from the fact that $f(\mcp_{i-1})-f(S)+f(\mcp')=f(\mcp_i)$ and $|\mcp'|=|\mcp_i|-|\mcp_{i-1}|+1$. Rearranging, we have that 
\begin{align}
    f(\mcp)&\leq \frac{f(\mcp_{i-1})}{|\mcp_i|-|\mcp_{i-1}|+1}(|\mcp_i|-k+1)+\frac{f(\mcp_{i})}{|\mcp_i|-|\mcp_{i-1}|+1}(k-|\mcp_{i-1}|)\notag\\
    &\quad \quad \quad \quad +\frac{f(S)}{|\mcp_i|-|\mcp_{i-1}|+1}(k-|\mcp_{i-1}|)\notag
    \\&=\left(\frac{|\mcp_i|-|\mcp_{i-1}|}{|\mcp_i|-|\mcp_{i-1}|+1}\right)\left(\frac{f(\mcp_{i-1})}{|\mcp_i|-|\mcp_{i-1}|}(|\mcp_i|-k+1)+\frac{f(\mcp_{i})+f(S)}{|\mcp_i|-|\mcp_{i-1}|}(k-|\mcp_{i-1}|)\right)\notag
    \\&\leq \left(\frac{|\mcp_i|-|\mcp_{i-1}|}{|\mcp_i|-|\mcp_{i-1}|+1}\right)\left(f(\mcp^\ast) + \frac{f(\mcp_{i-1})}{|\mcp_i|-|\mcp_{i-1}|} +\frac{f(S)}{|\mcp_i|-|\mcp_{i-1}|}(k-|\mcp_{i-1}|)\right)\notag\\ 
    &\leq \left(\frac{|\mcp_i|-|\mcp_{i-1}|}{|\mcp_i|-|\mcp_{i-1}|+1}\right)\left(f(\mcp^\ast) +\frac{f(\mcp_{i-1})}{|\mcp_i|-|\mcp_{i-1}|}(1+k-|\mcp_{i-1}|)\right).  \label{eqn: mono - 1}
\end{align}
where the second inequality above is by {by Lemma \ref{lemma: f long ineq}(i)} and the third inequality above is because $f(S)\leq f(\mcp_{i-1})$. 
Inequality \eqref{eqn: mono - 1} implies that 
\begin{align*}
    f(\mcp)&\leq \left(\frac{|\mcp_i|-|\mcp_{i-1}|}{|\mcp_i|-|\mcp_{i-1}|+1}\right)\left(f(\mcp^\ast)+f(\mcp_{i-1})+\frac{f(\mcp_{i-1})}{|\mcp_i|-|\mcp_{i-1}|}(1+k-|\mcp_i|)\right)
    \\&\leq \left(1-\frac{1}{|\mcp_i|-|\mcp_{i-1}|+1}\right)\left(f(\mcp^\ast)+f(\mcp_{i-1})\right)\quad(\text{since }k<|\mcp_i|)
    \\&\leq 2f(\mcp^\ast),
\end{align*}
where the last inequality is because $f(\mcp_{i-1})\leq f(\mcp^\ast)$ by Lemma \ref{lemma: f long ineq}(ii). 
The above analysis shows that the approximation factor is at most $2$. We tighten the analysis now. 
As a consequence of the above inequality, we may assume that $f(\mcp^\ast)\neq 0$ because if $f(\mcp^\ast)=0$, then the returned $k$-partition $\mcp$ also satisfies $f(\mcp)=0$ and thus, is optimal.
Let $c:=f(\mcp_{i-1})/f(\mcp^\ast)$. By Lemma \ref{lemma: f long ineq}(ii), we have that $f(\mcp_{i-1})\le f(\mcp^{\ast})$ and hence, $c\in [0, 1]$. For convenience, we define $A:=k-|\mcp_{i-1}|$ and $B:=|\mcp_i|-k$ and note that $A, B\ge 1$.
Using this notation, we may rewrite inequality \eqref{eqn: mono - 1} as
\begin{align}
    f(\mcp)&\leq \left(\frac{A+B}{A+B+1}\right)\left(f(\mcp^\ast)+\frac{1+A}{A+B}f(\mcp_{i-1})\right)\notag
    \\&=\left(\frac{A+B}{A+B+1}\right)\left(1+\frac{1+A}{A+B}\cdot c\right)f(\mcp^\ast).\label{eqn: 4/3 - 1}
\end{align}
By Lemma \ref{lemma: f long ineq}(i), we have
\[f(\mcp^\ast) \ge \left(\frac{B}{A+B}\right)f(\mcp_{i-1})+\left(\frac{A}{A+B}\right)f(\mcp_i) = \left(\frac{B}{A+B}\right)cf(\mcp^\ast)+\left(\frac{A}{A+B}\right)f(\mcp_i).\]
Rearranging, we have
\[f(\mcp_i)\leq \left(1-\frac{B}{A+B}\cdot c\right)\left(\frac{A+B}{A}\right)f(\mcp^\ast) = \left(\frac{A+B}{A}-\frac{B}{A}\cdot c\right) f(\mcp^\ast).\]
Since $\mcp$ is obtained by coarsening $\mcp_i$, we have $f(\mcp)\leq f(\mcp_i)$ by submodularity and non-negativity of $f$. This implies
\begin{align}
    f(\mcp)\leq 
    \left(\frac{A+B}{A}-\frac{B}{A}\cdot c\right) f(\mcp^\ast).\label{eqn: 4/3 - 2}
\end{align}
Combining inequalities \eqref{eqn: 4/3 - 1} and \eqref{eqn: 4/3 - 2}, we have
\begin{align}
    \frac{f(\mcp)}{f(\mcp^\ast)} \leq \max_{c\in[0,1]}\min\left\{\left(\frac{A+B}{A+B+1}\right)\left(1+\frac{1+A}{A+B}\cdot c\right),\; \frac{A+B}{A}-\frac{B}{A}\cdot c\right\}.\label{eqn:min-of-two-terms-monotone}
\end{align}
Thus, in order to upper bound the approximation factor, it suffices to upper bound the right hand side of inequality (\ref{eqn:min-of-two-terms-monotone}). 
Since $\left(\frac{A+B}{A+B+1}\right)\left(1+\frac{1+A}{A+B}\cdot c\right)$ and $\frac{A+B}{A}-\frac{B}{A}\cdot c$ are both linear in $c$, with the former increasing and the latter decreasing as a function of $c$, the value  $$\max_{c\in\R}\min\left\{\left(\frac{A+B}{A+B+1}\right)\left(1+\frac{1+A}{A+B}\cdot c\right),\; \frac{A+B}{A}-\frac{B}{A}\cdot c\right\}$$ is achieved when the two terms are equal. Setting $\left(\frac{A+B}{A+B+1}\right)\left(1+\frac{1+A}{A+B}\cdot c^\ast\right)=\frac{A+B}{A}-\frac{B}{A}\cdot c^\ast$ and solving for $c^\ast$, we get
\[c^\ast=\frac{\frac{A+B}{A}-\frac{A+B}{A+B+1}}{\frac{1+A}{A+B+1}+\frac{B}{A}}=\frac{\frac{B}{A}+\frac{1}{A+B+1}}{\frac{1+A}{A+B+1}+\frac{B}{A}}.\]
Plugging $c=c^\ast$ into $\frac{A+B}{A}-\frac{B}{A}\cdot c$ yields
\begin{align*}
    \frac{f(\mcp)}{f(\mcp^\ast)}&\leq \max_{c\in[0,1]}\min\left\{\frac{A+B}{A+B+1}\left(1+\frac{1+A}{A+B}\cdot c\right),\; \frac{A+B}{A}-\frac{B}{A}\cdot c\right\}
    \\&\leq \frac{A+B}{A}-\frac{B}{A}\cdot c^\ast
    \\&= 1+\frac{B}{A}(1-c^\ast)
    \\&=1+\frac{B}{A}\left(1-\frac{\frac{B}{A}+\frac{1}{A+B+1}}{\frac{1+A}{A+B+1}+\frac{B}{A}}\right)
    \\&=1+\frac{\frac{B}{A+B+1}}{\frac{1+A}{A+B+1}+\frac{B}{A}}
    \\&=1+\frac{AB}{A+A^2+AB+B^2+B}
    \\&\leq 1+\frac{AB}{3AB+A+B} \quad(\text{since }A^2+B^2\geq 2AB)
    \\&=\frac{4}{3}-\frac{1}{3}\cdot \frac{A+B}{3AB+A+B}
    \\&=\frac{4}{3}-\frac{1}{3}\cdot \frac{A+B}{3AB/2+3AB/2+A+B} 
    \\&\leq\frac{4}{3}-\frac{1}{3}\cdot \frac{A+B}{3AB/2+3(A^2+B^2)/4+A+B}\quad(\text{since }AB\leq \frac{A^2+B^2}{2})
    \\&=\frac{4}{3}-\frac{1}{3}\cdot \frac{A+B}{3(A+B)^2/4+A+B}
    \\&=\frac{4}{3}-\frac{1}{3}\cdot \frac{1}{3(A+B)/4+1}
    \\&\leq \frac{4}{3}-\frac{4}{9n+3}.
\end{align*}
The last inequality above is because $A+B=|\mcp_i|-|\mcp_{i-1}|\leq n-1$.
\end{proof}

\begin{remark}\label{remark:mono-4/3-tight}
The approximation factor of Algorithm 1 for non-negative monotone submodular functions is at least $4/3-4/(9n+3)$. We show this for $n=3$ using the following example:
Let $V=\{a,b,c\}$, $k=2$, and $f:2^V\to \R_{\ge 0}$ be defined by
\begin{align*}
    &f(\emptyset)=0,\;f(\{a\})=1,\;f(\{b\})=f(\{c\})=1+\epsilon,
    \\&f(\{a,b\}) = f(\{a,c\}) =\frac{3}{2}+\epsilon ,\;f(\{b,c\}) =f(V)= 2+2\epsilon.
\end{align*}
Submodularity and monotonicity of $f$ can be verified by considering all possible subsets. Moreover, the principal partition sequence of this instance is $\{V\},\{\{a\},\{b\},\{c\}\}$. Thus, Algorithm 1 returns the $2$-partition $\{\{a\},\{b,c\}\}$, whose objective value is $3+2\epsilon$. An optimum $2$-partition is $\{\{c\},\{a,b\}\}$, whose objective value is $5/2+\epsilon$. Thus, the approximation factor is  $\frac{3+2\epsilon}{5/2+\epsilon}\to \frac{6}{5}$ as $\epsilon\to 0$. We note that for $n=3$, the approximation factor guaranteed by Theorem \ref{Thm: 4/3-approx mono sm k-cut algo} is $\frac{4}{3}-\frac{4}{9n+3} = \frac{6}{5}$. 
\end{remark}

We conclude the section by showing that there exist monotone submodular functions for which the approximation factor of Algorithm 1 is at least $4/3$ asymptotically (i.e., as $n\to \infty$).
\begin{lemma}\label{lemma:monotone-tight-example}
For every odd positive integer $n\ge 5$, there exists a function $k=k(n)$ (i.e., $k$ is a function of $n$) and an instance of non-negative monotone submodular $k$-partition over an $n$-element ground set such that the approximation factor of Algorithm 1 on that instance is arbitrarily close to $4/3-4/(3n+3)$. 
\end{lemma}

\begin{proof}
Let $n\geq 5$ be an arbitrary odd number, $k=\frac{n+1}{2}$, and let $V=\{v_1,\ldots, v_n\}$ be the ground set. Moreover, let $U:=\{v_1,\ldots, v_{\frac{n-1}{2}}\}$ and $D:=\{v_{\frac{n+1}{2}},\ldots, v_n\}$ so that $V=U\uplus D$. Let $g:2^U\to \R_{\geq 0}$ be a function over the ground set $U$ defined by 
\[g(S) = \begin{cases}\frac{1}{2}+\frac{1}{2}\cdot |S| & \ \text{ if }\ \emptyset \neq S\subseteq U, \\ 0 & \ \text{ if } S=\emptyset.\end{cases}\]
and $f:2^V\to \R_{\geq 0}$ be defined by
\[f(S):=\min\left\{g(S\cap U)+(1+\epsilon)|S\cap D|,\;\frac{n+1}{2}\right\}\ \forall\ S\subseteq V,\]
where $\epsilon>0$ is infinitesimally small.  
The function $f$ satisfies $f(\emptyset) = 0$, $f(V)=\frac{n+1}{2}$, $f(U) = \frac{n+1}{4}$, $f(D)=\frac{n+1}{2}$, $f(\{v\})=1$ for each $v\in U$, and $f(\{v\})=1+\epsilon$ for each $u\in D$. We will use $\mcq$ to denote the partition of $V$ into $n$ singleton sets, and it follows that $f(\mcq) = n+\frac{(n+1)\epsilon}{2}$. For convenience, we will write $h(S):=g(S\cap U)+|S\cap D|$ for all $S\subseteq V$ so that $f(S) = \min\{h(S),\frac{n+1}{2}\}$ for all $S\subseteq V$ throughout the proof.

\begin{claim}
The function $f$ is submodular and monotone.
\end{claim}
\begin{proof}
Let $A,B\subseteq V$ be arbitrary. 
We notice that the function $h$ is monotone, and
\[h(A)+h(B)\geq h(A\cap B)+h(A\cup B).\]
Without loss of generality, we may assume $h(A)\geq h(B)$, and thus $h(A\cup B)\geq h(A)\geq h(B)\geq h(A\cap B)$. We consider the following five cases:
\begin{enumerate}
    \item $h(A\cup B)\leq \frac{n+1}{2}$. 
    \[\implies f(A)+f(B) = h(A)+h(B)\geq h(A\cap B)+h(A\cup B) = f(A\cap B)+f(A\cup B).\]
    \item $h(A)\leq \frac{n+1}{2}<h(A\cup B)$.
    \[\implies f(A)+f(B) = h(A)+h(B)\geq h(A\cap B)+h(A\cup B) > h(A\cap B)+\frac{n+1}{2} = f(A\cap B)+f(A\cup B).\]
    \item $h(B)\leq \frac{n+1}{2}<h(A)$.
    \[\implies f(A)+f(B) = \frac{n+1}{2}+h(B)\geq \frac{n+1}{2}+h(A\cap B) = f(A\cup B)+f(A\cap B).\]
    \item $h(A\cap B)\leq \frac{n+1}{2}<h(B)$.
    \[\implies f(A)+f(B) = \frac{n+1}{2}+\frac{n+1}{2}\geq \frac{n+1}{2}+h(A\cap B) = f(A\cup B)+f(A\cap B).\]
    \item $\frac{n+1}{2}<h(A\cap B)$.
    \[\implies f(A)+f(B) = \frac{n+1}{2}+\frac{n+1}{2} = f(A\cup B)+f(A\cap B).\]
\end{enumerate}
Since exactly one of these five cases holds, we have proved the submodularity of $f$. The monotonicity of $f$ is implied by the monotonicity of $h$.
\end{proof}
\begin{claim}\label{claim: 4/3 tight example - 2}
For every partition $\mcp\neq\mcq, \{V\}$, the function $f$ satisfies
\[\frac{f(\mcp)-f(V)}{|\mcp|-1}> \frac{f(\mcq)-f(V)}{n-1}.\]
\end{claim}
\begin{proof}
First, we note that $f(V)=\frac{n+1}{2}$ and $\frac{f(\mcq)-f(V)}{n-1} = \frac{n+(n+1)\epsilon/2-(n+1)/2}{n-1}=\frac{1}{2}+\frac{(n+1)\epsilon}{2n-2}$, and thus the desired inequality is equivalent to  
\[f(\mcp)\geq \frac{n}{2}+\frac{1}{2}|\mcp|+\epsilon(|\mcp|-1)\frac{n+1}{2n-2}.\]
We note that there exists at most one part $P\in\mcp$ such that $h(P)\geq \frac{n+1}{2}$. To see this, assume that two distinct parts $P,P'\in\mcp$ satisfy $h(P)\geq\frac{n+1}{2}$ and $h(P')\geq\frac{n+1}{2}$. This implies
\begin{align*}
    n+1&\leq h(P)+h(P')\leq 2\cdot\frac{1}{2}+\frac{1}{2}|(P\cup P')\cap U|+(1+\epsilon)|(P\cup P')\cap D|
    \\&\leq 1+\frac{1}{2}\cdot \frac{n-1}{2}+(1+\epsilon)\frac{n+1}{2}=\frac{3n+1}{4}+1+\frac{(n+1)\epsilon}{2}<n+1,
\end{align*}
because $\epsilon$ is infinitesimal and $n\geq 5$, yielding a contradiction. Therefore, we may consider two cases: $\mcp$ containing no part $P$ such that $h(P)\geq\frac{n+1}{2}$, and $\mcp$ containing exactly one part $P$ such that $h(P)\geq \frac{n+1}{2}$.

Suppose $\mcp$ contains no part $P$ such that $h(P)\geq\frac{n+1}{2}$. We will use $t$ to denote the number of parts of $\mcp$ that each intersects $U$ non-trivially. Then, $|\mcp|-t$ represents the number of parts of $\mcp$ that are each contained in $D$, and we have $|\mcp|-t\leq |D|=\frac{n+1}{2}$. It follows that
\begin{align*}
    f(\mcp)-\frac{n}{2}-\frac{|\mcp|}{2}-\epsilon(|\mcp|-1)\frac{n+1}{2n-2}
    &=\frac{1}{2}\cdot t+\frac{1}{2}\cdot\frac{n-1}{2}+(1+\epsilon)\frac{n+1}{2}\\
    & \quad \quad \quad \quad \quad \quad -\frac{n}{2}-\frac{|\mcp|}{2}-\epsilon(|\mcp|-1)\frac{n+1}{2n-2}
    \\&=\frac{n+1}{4}-\frac{|\mcp|-t}{2}+\epsilon\left(\frac{n+1}{2}-\frac{(|\mcp|-1)(n+1)}{2n-2}\right)
    \\&\geq \frac{n+1}{4}-\frac{n+1}{4}+ \epsilon\left(\frac{n+1}{2}-\frac{(|\mcp|-1)(n+1)}{2n-2}\right)
    \\&\geq\epsilon\left(\frac{n+1}{2}-\frac{(n-2)(n+1)}{2n-2}\right) \ \ (\text{since }|\mcp|\leq n-1)
    \\& = \epsilon\cdot\frac{n+1}{2n-2}
    \\&>0.
\end{align*}
Now, suppose $\mcp$ contains exactly one part $P\in\mcp$ such that $h(P)\geq\frac{n+1}{2}$. Each other part $P'\neq P$ of $\mcp$ satisfies $f(P')\geq 1$. It follows that 
\begin{align*}
    f(\mcp)\geq \frac{n+1}{2}+1\cdot (|\mcp|-1)=\frac{n}{2}+|\mcp|-\frac{1}{2}> \frac{n}{2}+\frac{1}{2}\cdot |\mcp|+\epsilon(|\mcp|-1)\frac{n+1}{2n-2}
\end{align*}
since $|\mcp|\geq 2$ and $\epsilon$ is arbitrarily small.
\end{proof}

By Proposition \ref{prop:unique-principal-partition}, Algorithm 1 returns the partition $\{\{v_1\},\ldots, \{v_{\frac{n-1}{2}}\}, D\}$ (because $\{v_1\},\ldots, \{v_{\frac{n-1}{2}}\}$ are the $k-1$ singleton sets that minimize the $f$ values), whose objective is $\frac{n-1}{2}+\frac{n+1}{2} = n$. The partition $\{\{v_1,\ldots, v_{\frac{n+1}{2}}\}, \{v_{\frac{n+3}{2}}\},\ldots, \{v_n\}\}$ has objective $\frac{n+1}{4}+(1+\epsilon)+(1+\epsilon)\frac{n-1}{2}=\frac{3n+3}{4}+\frac{(n+1)\epsilon}{2}$. The approximation factor is at least
\[\frac{n}{\frac{3n+3}{4}+\frac{(n+1)\epsilon}{2}}\to\frac{4}{3}-\frac{4}{3n+3} \;(\text{as }\epsilon\to 0).\]
This completes the proof of Lemma \ref{lemma:monotone-tight-example}. 
\end{proof}

\subsection{Symmetric submodular functions}\label{sec:sym}
In this section, we bound the approximation factor of Algorithm 1 for symmetric submodular $k$-partitioning. The following is the main theorem of this section. 


\begin{theorem} \label{Thm: 2-approx sym sm k-cut algo}
The approximation factor of Algorithm 1 for non-negative symmetric submodular $k$-partitioning is $2(1-1/n)$, where $n$ is the size of the ground set. 
\end{theorem}

The asymptotic approximation factor of $2$ achieved by Algorithm 1 is the best possible for non-negative symmetric submodular $k$-partition: for every constant $\epsilon>0$, there does not exist an algorithm that achieves a $(2-\epsilon)$-approximation using polynomial number of function evaluation queries \cite{San21}.  
We will remark on the tightness of the approximation factor for Algorithm 1 after proving the theorem. 
The proof of Theorem  \ref{Thm: 2-approx sym sm k-cut algo} follows from Lemma \ref{claim:exact-k} and Lemma \ref{claim:not-exact-k} shown below.
\begin{lemma}\label{claim:not-exact-k}
Let $k\geq 2$ be an integer, $f:2^V\to\R$ be a non-negative symmetric submodular function on a ground set $V$ of size $n$, $\mcp^\ast$ be a $k$-partition of $V$ that minimizes $f(\mcp^\ast)$, and $\mcp_1, \mcp_2, \ldots, \mcp_r$ be a principal partition sequence of the submodular function $f$ satisfying the conditions of Theorem \ref{Thm: PPS}. 
Suppose $|\mcp_j|\neq k$ for all $j\in[r]$.
Then the partition $\mcp$ returned by Algorithm 1
satisfies 
\[f(\mcp)\leq 2\left(1-\frac{1}{n}\right)f(\mcp^\ast).\]
\end{lemma}

\begin{proof}
Let $\mcp_{i-1},\mcp_i$ be the partitions such that $|\mcp_{i-1}|<k<|\mcp_i|$. Let $S$ and $\mcp'=\{B_1,B_2,\ldots,B_{|\mcp'|}\}$ be as in Algorithm 1. 

Firstly, we note that for every $T\subseteq S$, symmetry, submodularity, and non-negativity of $f$ together imply that
\begin{align}
    f(T)&=f(V-T)=f((V-S)\cup(S-T))
    \leq f(V-S)+f(S-T)=f(S)+f(S-T). \label{Eqn: sym sm}
\end{align}
Secondly, by our choice of $B_1,B_2,\ldots,B_{k-|\mcp_{i-1}|}$, we know that
\begin{align}
    \sum_{j=1}^{k-|\mcp_{i-1}|}f(B_j)\leq\frac{k-|\mcp_{i-1}|}{|\mcp'|}f(\mcp').\label{Eqn: sum of f(Bj)}
\end{align}
Since $\cup_{j=k-|\mcp_{i-1}|+1}^{|\mcp'|}B_j\subseteq S$, we have that 
\begin{align*}
    f\left(\bigcup_{j=k-|\mcp_{i-1}|+1}^{|\mcp'|}B_j\right)
    &\leq f(S)+f\left(\bigcup_{j=1}^{k-|\mcp_{i-1}|}B_j\right)\quad\text{(by inequality \eqref{Eqn: sym sm})}
    \\&\leq f(S)+\sum_{j=1}^{k-|\mcp_{i-1}|}f(B_j)\quad\text{(by submodularity and non-negativity)}
    \\&\leq f(S)+\frac{k-|\mcp_{i-1}|}{|\mcp'|}f(\mcp').\quad\text{(by inequality }\eqref{Eqn: sum of f(Bj)})
\end{align*}
Therefore, we have
\begin{align*}
    f(\mcp)&=f(\mcp_{i-1})-f(S)+\sum_{j=1}^{k-|\mcp_{i-1}|}f(B_j)+f\left(\bigcup_{j=k-|\mcp_{i-1}|+1}^{|\mcp'|}B_j\right)
    \\&\leq f(\mcp_{i-1})-f(S)+\frac{k-|\mcp_{i-1}|}{|\mcp'|}f(\mcp')+\left(f(S)+\frac{k-|\mcp_{i-1}|}{|\mcp'|}f(\mcp')\right)
    \\&=f(\mcp_{i-1})+2\cdot\frac{k-|\mcp_{i-1}|}{|\mcp'|}f(\mcp')
    \\&= f(\mcp_{i-1})+2\cdot\frac{k-|\mcp_{i-1}|}{|\mcp_i|-|\mcp_{i-1}|+1}(f(\mcp_i)+f(S)-f(\mcp_{i-1})),
\end{align*}
where the last equality follows from the fact that $f(\mcp_{i-1})-f(S)+f(\mcp')=f(\mcp_i)$ and $|\mcp'|=|\mcp_i|-|\mcp_{i-1}|+1$. Rearranging, we get
\begin{align*}
    f(\mcp)
    \leq \frac{f(\mcp_{i-1})}{|\mcp_i|-|\mcp_{i-1}|+1}(|\mcp_i|+|\mcp_{i-1}|-2k+1)&+\frac{f(\mcp_{i})}{|\mcp_i|-|\mcp_{i-1}|+1}(2k-2|\mcp_{i-1}|)\\
    & \quad \quad \quad \quad +\frac{k-|\mcp_{i-1}|}{|\mcp_i|-|\mcp_{i-1}|+1}\cdot 2f(S).
\end{align*}
Now, we observe that $2f(S)=f(S)+f(V-S)\leq f(\mcp_{i-1})$ by symmetry, submodularity, and non-negativity. Consequently, we have that 
\begin{align*}
    f(\mcp)&\leq \frac{f(\mcp_{i-1})}{|\mcp_i|-|\mcp_{i-1}|+1}(|\mcp_i|+|\mcp_{i-1}|-2k+1)+\frac{f(\mcp_{i})}{|\mcp_i|-|\mcp_{i-1}|+1}(2k-2|\mcp_{i-1}|)\\
    & \quad \quad \quad \quad \quad \quad \quad \quad +\frac{k-|\mcp_{i-1}|}{|\mcp_i|-|\mcp_{i-1}|+1}\cdot f(\mcp_{i-1})
    \\&=\left(\frac{|\mcp_i|-|\mcp_{i-1}|}{|\mcp_i|-|\mcp_{i-1}|+1}\right)\left(\frac{f(\mcp_{i-1})}{|\mcp_i|-|\mcp_{i-1}|}(|\mcp_i|-k+1)+\frac{f(\mcp_{i})}{|\mcp_i|-|\mcp_{i-1}|}(2k-2|\mcp_{i-1}|)\right)
    \\&\le\left(\frac{|\mcp_i|-|\mcp_{i-1}|}{|\mcp_i|-|\mcp_{i-1}|+1}\right)\left(\frac{f(\mcp_{i-1})}{|\mcp_i|-|\mcp_{i-1}|}(2|\mcp_i|-2k)+\frac{f(\mcp_{i})}{|\mcp_i|-|\mcp_{i-1}|}(2k-2|\mcp_{i-1}|)\right) 
    \\&\leq \left(\frac{|\mcp_i|-|\mcp_{i-1}|}{|\mcp_i|-|\mcp_{i-1}|+1}\right)\cdot 2f(\mcp^\ast)\quad(\text{by Lemma \ref{lemma: f long ineq}(i)})
    \\&= \left(1-\frac{1}{|\mcp_i|-|\mcp_{i-1}|+1}\right)\cdot 2f(\mcp^\ast)
    \\&\leq 2\left(1-\frac{1}{n}\right)f(\mcp^\ast),
\end{align*}
where the second inequality is because $1\le |\mcp_i|-k$ and the last inequality is because $|\mcp_i|-|\mcp_{i-1}|\leq n-1$.
\end{proof}

\begin{remark}\label{remark:lower-bound-example-for-sym-sm-fns}
In contrast to the greedy splitting and the Gomory-Hu tree based approaches, 
the principal partition sequence approach given in Algorithm 1 
does not achieve an approximation factor of $2-2/k$ for non-negative symmetric submodular functions: For the special case of graph $k$-cut, Proposition 5.2 of \cite{RS07} provides a family of graph cut functions on which Algorithm 1 returns solutions whose approximation factor is arbitrarily close to $2$. 
\end{remark}




\subsection{Posimodular submodular functions}\label{sec:pm}
In this section, we bound the approximation factor of Algorithm 1 for posimodular submodular $k$-partitioning. The following is the main theorem of this section. 

\begin{theorem} \label{Thm: 2-approx pm sm k-cut algo}
The approximation factor of Algorithm 1 for non-negative posimodular submodular $k$-partitioning is $2(1-\frac{1}{n+1})$, where $n$ is the size of the ground set. 
\end{theorem}

The asymptotic approximation factor of $2$ achieved by Algorithm 1 is the best possible for non-negative posimodular submodular $k$-partition: for every constant $\epsilon>0$, there does not exist an algorithm that achieves a $(2-\epsilon)$-approximation using polynomial number of function evaluation queries \cite{San21}.  
Lemmas \ref{claim:exact-k} and \ref{lemma:posimodular} complete the proof of Theorem \ref{Thm: 2-approx pm sm k-cut algo}. 
The proof of Theorem  \ref{Thm: 2-approx pm sm k-cut algo} follows from Lemma \ref{claim:exact-k} and Lemma \ref{lemma:posimodular} shown below.

\begin{lemma}\label{lemma:posimodular}
Let $k\geq 2$ be an integer, $f:2^V\to\R_{\ge 0}$ be a non-negative posimodular submodular function on a ground set $V$ of size $n$, $\mcp^\ast$ be a $k$-partition of $V$ that minimizes $f(\mcp^\ast)$, and $\mcp_1, \mcp_2, \ldots, \mcp_r$ be a principal partition sequence of the submodular function $f$ satisfying the conditions of Theorem \ref{Thm: PPS}. 
Suppose $|\mcp_j|\neq k$ for all $j\in[r]$.
Then, the partition $\mcp$ returned by Algorithm 1
satisfies 
\[f(\mcp)\leq 2\left(1-\frac{1}{n+1}\right)f(\mcp^\ast).\]
\end{lemma}

\begin{proof}
Let $\mcp_{i-1},\mcp_i$ be the partitions such that $|\mcp_{i-1}|<k<|\mcp_i|$. Let $S$ and $\mcp'=\{B_1,B_2,\ldots,B_{|\mcp'|}\}$ be as in Algorithm 1. 

Firstly, we note that for every $T\subseteq S$, non-negativity and posimodularity implies that
\begin{align}
f(T)\le f(T) + f(\emptyset) \le f(S) + f(S-T). \label{Eqn: posi sm}
\end{align}
Secondly, by our choice of $B_1,B_2,\ldots,B_{k-|\mcp_{i-1}|}$, we know that
\begin{align}
    \sum_{j=1}^{k-|\mcp_{i-1}|}f(B_j)\leq\frac{k-|\mcp_{i-1}|}{|\mcp'|}f(\mcp').\label{Eqn: sum of f(Bj)-posi}
\end{align}
Since $\cup_{j=k-|\mcp_{i-1}|+1}^{|\mcp'|}B_j\subseteq S$, we have that 
\begin{align*}
    f\left(\bigcup_{j=k-|\mcp_{i-1}|+1}^{|\mcp'|}B_j\right)
    &\leq f(S)+f\left(\bigcup_{j=1}^{k-|\mcp_{i-1}|}B_j\right)\quad\text{(by inequality \eqref{Eqn: posi sm})}
    \\&\leq f(S)+\sum_{j=1}^{k-|\mcp_{i-1}|}f(B_j)\quad\text{(by submodularity and non-negativity)}
    \\&\leq f(S)+\frac{k-|\mcp_{i-1}|}{|\mcp'|}f(\mcp').\quad\text{(by inequality }\eqref{Eqn: sum of f(Bj)-posi})
\end{align*}
Therefore, we have 
\begin{align*}
    f(\mcp)
    &=f(\mcp_{i-1})-f(S)+\sum_{j=1}^{k-|\mcp_{i-1}|}f(B_j)+f\left(\bigcup_{j=k-|\mcp_{i-1}|+1}^{|\mcp'|}B_j\right)
    \\&\leq f(\mcp_{i-1})-f(S)+2\sum_{j=1}^{k-|\mcp_{i-1}|}f(B_j)+f(S)
    \\&= f(\mcp_{i-1})+2\cdot\sum_{j=1}^{k-|\mcp_{i-1}|}f(B_j)
    \\&\leq f(\mcp_{i-1})+2\cdot\frac{k-|\mcp_{i-1}|}{|\mcp'|}f(\mcp')
    \\&=f(\mcp_{i-1})+\frac{2k-2|\mcp_{i-1}|}{|\mcp_i|-|\mcp_{i-1}|+1}(f(\mcp_i)+f(S)-f(\mcp_{i-1})),
\end{align*}
where the last equality follows from the fact that $f(\mcp_{i-1})-f(S)+f(\mcp')=f(\mcp_i)$ and $|\mcp'|=|\mcp_i|-|\mcp_{i-1}|+1$. 
Applying the fact that $f(S)\leq f(\mcp_{i-1})$, we have
\begin{align}
    f(\mcp)
    &\le \left(\frac{|\mcp_i|-|\mcp_{i-1}|}{|\mcp_i|-|\mcp_{i-1}|+1}\right)\left(\frac{f(\mcp_{i-1})}{|\mcp_i|-|\mcp_{i-1}|}(|\mcp_i|-|\mcp_{i-1}|+1)\right)\notag\\
    & \quad \quad \quad \quad \quad \quad +\left(\frac{|\mcp_i|-|\mcp_{i-1}|}{|\mcp_i|-|\mcp_{i-1}|+1}\right)\left(\frac{f(\mcp_{i})}{|\mcp_i|-|\mcp_{i-1}|}(2k-2|\mcp_{i-1}|)\right) \notag
    \\&=\left(\frac{|\mcp_i|-|\mcp_{i-1}|}{|\mcp_i|-|\mcp_{i-1}|+1}\right)\left(\frac{f(\mcp_{i-1})}{|\mcp_i|-|\mcp_{i-1}|}((2|\mcp_i|-2k)+(2k-|\mcp_{i-1}|-|\mcp_i|+1))\right)\notag\\
    & \quad \quad \quad \quad \quad \quad +\left(\frac{|\mcp_i|-|\mcp_{i-1}|}{|\mcp_i|-|\mcp_{i-1}|+1}\right)\left(\frac{2f(\mcp_{i})}{|\mcp_i|-|\mcp_{i-1}|}(k-|\mcp_{i-1}|)\right) \notag
    \\&\leq \left(\frac{|\mcp_i|-|\mcp_{i-1}|}{|\mcp_i|-|\mcp_{i-1}|+1}\right)\left(2f(\mcp^\ast)+\frac{f(\mcp_{i-1})}{|\mcp_i|-|\mcp_{i-1}|}(2k-|\mcp_{i-1}|-|\mcp_i|+1)\right), \label{eqn: posimodular - 1}
\end{align}
where the last inequality is by Lemma \ref{lemma: f long ineq}(i). Inequality \eqref{eqn: posimodular - 1} implies that 
\begin{align*}
    f(\mcp)&\leq \left(\frac{|\mcp_i|-|\mcp_{i-1}|}{|\mcp_i|-|\mcp_{i-1}|+1}\right)\left(2f(\mcp^\ast)+f(\mcp_{i-1})+\frac{f(\mcp_{i-1})}{|\mcp_i|-|\mcp_{i-1}|}(2k-2|\mcp_i|+1)\right)
    \\&\leq \left(1-\frac{1}{|\mcp_i|-|\mcp_{i-1}|+1}\right)\left( 2f(\mcp^\ast)+f(\mcp_{i-1})\right)\quad(\text{since }k<|\mcp_i|)
    \\&\leq 3f(\mcp^\ast),
\end{align*}
where the last inequality is because $f(\mcp_{i-1})\leq f(\mcp^\ast)$ by Lemma \ref{lemma: f long ineq}(ii). 
The above analysis already shows that the approximation factor is at most $3$. We tighten the factor now. 
As a consequence of the above inequality, we may assume that $f(\mcp^\ast)\neq 0$ because if $f(\mcp^\ast)=0$, then the returned $k$-partition $\mcp$ also satisfies $f(\mcp)=0$ and thus, is optimal.
Let $c:=f(\mcp_{i-1})/f(\mcp^\ast)$. By Lemma \ref{lemma: f long ineq}(ii), we have that $f(\mcp_{i-1})\le f(\mcp^{\ast})$ and hence, $c\in [0, 1]$. For convenience, we define $A:=k-|\mcp_{i-1}|$ and $B:=|\mcp_i|-k$ and note that $A, B\ge 1$.
Using these notations, we may rewrite inequality \eqref{eqn: posimodular - 1} as
\begin{align}
    f(\mcp)&\leq \left(\frac{A+B}{A+B+1}\right)\left(2f(\mcp^\ast)+\frac{A-B+1}{A+B}f(\mcp_{i-1})\right) \notag\\
    &= \left(\frac{A+B}{A+B+1}\right)\left(2+\frac{A-B+1}{A+B}\cdot c\right)f(\mcp^\ast).\label{eqn: posimodular - 2}
\end{align}
By Lemma \ref{lemma: f long ineq}(i), we have
\[f(\mcp^\ast) \ge \left(\frac{B}{A+B}\right)f(\mcp_{i-1})+\left(\frac{A}{A+B}\right)f(\mcp_i) = \left(\frac{B}{A+B}\right)cf(\mcp^\ast)+\left(\frac{A}{A+B}\right)f(\mcp_i).\]
Rearranging, we have
\[f(\mcp_i)\leq \left(1-\frac{B}{A+B}\cdot c\right)\left(\frac{A+B}{A}\right)f(\mcp^\ast) = \left(\frac{A+B}{A}-\frac{B}{A}\cdot c\right) f(\mcp^\ast).\]
Since $\mcp$ is obtained by coarsening $\mcp_i$, we have $f(\mcp)\leq f(\mcp_i)$ by submodularity and non-negativity of $f$. Hence, 
\begin{align}
    f(\mcp)\leq 
    \left(\frac{A+B}{A}-\frac{B}{A}\cdot c\right) f(\mcp^\ast).\label{eqn: posimodular - 3}
\end{align}
Combining inequalities \eqref{eqn: posimodular - 2} and \eqref{eqn: posimodular - 3}, we have
\begin{align}
    \frac{f(\mcp)}{f(\mcp^\ast)} \leq \max_{c\in[0,1]}\min\left\{\left(\frac{A+B}{A+B+1}\right)\left(2+\frac{A-B+1}{A+B}\cdot c\right),\; \frac{A+B}{A}-\frac{B}{A}\cdot c\right\}. \label{eqn:min-of-two-terms-posimodular}
\end{align}
Thus, in order to upper bound the approximation factor, it suffices to upper bound the right hand side of inequality (\ref{eqn:min-of-two-terms-posimodular}). 
Both terms $\left(\frac{A+B}{A+B+1}\right)\left(2+\frac{A-B+1}{A+B}\cdot c\right)$ and $\frac{A+B}{A}-\frac{B}{A}\cdot c$ are linear in $c$, and the latter is decreasing as a function of $c$. Next, we consider two cases: $A-B+1\leq 0$ and $A-B+1>0$.

Suppose $A-B+1\leq 0$, then the term $\left(\frac{A+B}{A+B+1}\right)\left(2+\frac{A-B+1}{A+B}\cdot c\right)$ is linear and non-increasing as a function of $c$. The maximum $$\max_{c\in[0,1]}\min\left\{\left(\frac{A+B}{A+B+1}\right)\left(2+\frac{A-B+1}{A+B}\cdot c\right),\; \frac{A+B}{A}-\frac{B}{A}\cdot c\right\}$$ is achieved at $c=0$. Thus, we have
\begin{align}
    \frac{f(\mcp)}{f(\mcp^\ast)}
    &\leq \min\left\{\frac{A+B}{A+B+1}\cdot 2,\; \frac{A+B}{A}\right\}\notag\\
    &\leq \frac{A+B}{A+B+1}\cdot 2 \notag\\
    &= 2\left(1-\frac{1}{A+B+1}\right)\notag\\
    &\leq 2\left(1-\frac{1}{n}\right), \label{eqn: posimodular - 4}
\end{align}
where the last inequality follows from the fact that $A+B=|\mcp_i|-|\mcp_{i-1}|\leq n-1$.

Now, we consider the case $A-B+1>0$. The term $\left(\frac{A+B}{A+B+1}\right)\left(2+\frac{A-B+1}{A+B}\cdot c\right)$ is linear and increasing as a function of $c$, and thus the maximum $$\max_{c\in\R}\min\left\{\left(\frac{A+B}{A+B+1}\right)\left(2+\frac{A-B+1}{A+B}\cdot c\right),\; \frac{A+B}{A}-\frac{B}{A}\cdot c\right\}$$ is achieved when the two terms are equal. Setting $\left(\frac{A+B}{A+B+1}\right)\left(2+\frac{A-B+1}{A+B}\cdot c^\ast\right) = \frac{A+B}{A}-\frac{B}{A}\cdot c^\ast$ and solving for $c^\ast$, we get
\begin{align*}
    c^\ast= \frac{\frac{A+B}{A}-2\cdot \frac{A+B}{A+B+1}}{\frac{A-B+1}{A+B+1}+\frac{B}{A}}.
\end{align*}
Plugging $c=c^\ast$ into $\frac{A+B}{A}-\frac{B}{A}\cdot c$, we have
\begin{align*}
    \frac{f(\mcp)}{f(\mcp^\ast)} &\leq\max_{c\in[0,1]}\min\left\{\left(\frac{A+B}{A+B+1}\right)\left(2+\frac{A-B+1}{A+B}\cdot c\right),\; \frac{A+B}{A}-\frac{B}{A}\cdot c\right\}
    \\&\leq \frac{A+B}{A}-\frac{B}{A}\cdot c^\ast
    \\&= 1+\frac{B}{A}(1-c^\ast)
    \\&=1+\frac{B}{A}\left(1-\frac{\frac{A+B}{A}-2\cdot \frac{A+B}{A+B+1}}{\frac{A-B+1}{A+B+1}+\frac{B}{A}}\right)
    \\&=1+\frac{\frac{2B}{A+B+1}}{\frac{A-B+1}{A+B+1}+\frac{B}{A}}
    \\&=1+\frac{2AB}{A^2+A+B^2+B}
    \\&\leq 1+\frac{2AB}{2AB+A+B}\quad(\text{since }A^2+B^2\geq 2AB)
    \\&=2-\frac{A+B}{2AB+A+B}
    \\&= 2-\frac{A+B}{AB+AB+A+B}
    \\&\leq 2-\frac{A+B}{AB+(A^2+B^2)/2+A+B}\quad(\text{since }AB\leq \frac{A^2+B^2}{2})
    \\&=2-\frac{A+B}{(A+B)^2/2+A+B}
    \\&=2-\frac{1}{(A+B)/2+1}
    \\&\leq 2-\frac{2}{n+1}.
\end{align*}
The last inequality above is because $A+B=|\mcp_i|-|\mcp_{i-1}|\leq n-1$.
\end{proof}


\begin{remark}
The approximation factor of Algorithm 1 for non-negative posimodular submodular functions is at least $2-2/(n+1)$.  
We show this for $n=3, k=2$ using the following example: Let $V=\{a, b,c\}$, $k=2$, and $f:2^V\rightarrow \R_{\ge 0}$ be defined by 
\begin{align*}
    &f(\emptyset) = 0, \;f(a) = f(b) = 1,\; f(c) = 1+\epsilon,
    \\& f(\{a,b\}) = 1+\epsilon,\; f(\{b,c\})=f(\{a,c\}) = 2,\; f(\{a,b,c\}) = 1+\epsilon.
\end{align*}
Submodularity and posimodularity of $f$ can be verified by considering all possible subsets. Moreover, the principal partition sequence of this instance is $\{V\},\{\{a\},\{b\},\{c\}\}$. Thus, the algorithm returns the $2$-partition $\{\{a\}, \{b,c\}\}$, while the optimum $2$-partition is $\{\{c\}, \{a,b\}\}$. Thus, the approximation factor approaches $3/2$ as $\epsilon\to 0$. We note that for $n=3$, the approximation factor guaranteed by Theorem \ref{Thm: 2-approx pm sm k-cut algo} is $2-2/(n+1)=3/2$. 
\end{remark}

\section{Lower bound for arbitrary submodular functions} 
\label{subsec: PPS n/k-apx}
In this section, we present an instance of submodular $k$-partition where Algorithm 1 achieves an approximation factor of $\Omega(n/k)$. We emphasize that the submodular function in our instance is not  symmetric/monotone/posimodular. 

Let $V=\{v_0, v_1, \ldots, v_{n-1}\}$ be the ground set. We define a digraph $D=(V,E(D))$ and a hypergraph $H=(V, E(H))$ on the same vertex set $V$ as follows (see Fig.\ref{fig: PPS n/k-apx}):
\begin{align*}
    E(D) &= \{v_0v_i:i\in[n-1]\} \text{ and }
    \\E(H) &= \{\{v_1, v_2,\ldots, v_{n-1}\}\}. 
\end{align*}

\begin{figure}[H]
    \centering
    \includegraphics[width=0.25\textwidth]{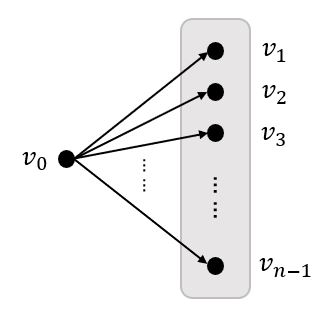}
    \caption{Example in Section \ref{subsec: PPS n/k-apx}. The arcs belong to the digraph $D$ and the hyperedge $\{v_1,\ldots, v_{n-1}\}$ belongs to the hypergraph $H$.}
    \label{fig: PPS n/k-apx}
\end{figure}

For every subset $S\subseteq V$, we will use $d_D^{in}(S)$ to denote the number of arcs in $D$ whose tails are in $\bar{S}$ and heads are in $S$. We will use $d_H(S)$ to denote the number of hyperedges in $H$ that have at least one vertex in $S$ and one vertex in $\bar{S}$. 
Next, we define a set function $f:2^V\to\R_{\geq 0}$ by
\[f(S):=a\cdot d^{in}_D(S)+d_H(S)\quad\forall S\subseteq V,\]
where $a\gg 1$ is a large constant. We note that $f$ is submodular because it is a positive linear combination of two submodular functions (and it is not monotone/symmetric/posimodular).

\begin{claim}\label{claim: PPS n/k-apx}
The principal partition sequence of $f$ is $\{V\}, \{\{v_i\}:i\in\{0,1,\ldots, n-1\}\}$.
\end{claim}

\begin{proof}
For convenience, we will use $\mcq$ to denote the partition of $V$ into singletons.
By Proposition \ref{prop:unique-principal-partition} and the fact that $f(V)=0$, it suffices to prove that for every partition $\mcp$ of $V$ such that $\mcp$ is not $\mcq$ or $\{V\}$, we have that 
\begin{align}
    \frac{f(\mcp)}{|\mcp|-1}>\frac{f(\mcq)}{n-1}.\label{eqn: PPS n/k-apx -2}
\end{align}
Let $P_0\in \mcp$ be the part that contains $v_0$. Then, we have that $f(P_0)\ge 1$ if $P_0\neq \{v_0\}$ and $f(P_0)=0$ otherwise. For each part $P\in \mcp$ that does not contain $v_0$, we have that $f(P)\ge 1+a$ if $|P|=1$ and $f(P)\ge 2+a$ if $|P|\ge 2$. Since $\mcp \neq \mcq$, we have that either $P_0\neq \{v_0\}$ or at least one of the parts $P\in \mcp\setminus \{P_0\}$ has size $|P|\ge 2$. Thus, $f(\mcp)=\sum_{P\in \mcp}f(P)\ge (1+a)(|\mcp|-1)+1$.
Moreover, we have $f(\mcq)=(1+a)(n-1)$ and hence
\begin{align*}
    \frac{f(\mcp)}{|\mcp|-1}&\geq\frac{(1+a)(|\mcp|-1)+1}{|\mcp|-1}=1+a+\frac{1}{|\mcp|-1}
    >1+a=\frac{(1+a)(n-1)}{n-1}=\frac{f(\mcq)}{n-1}.
\end{align*}
This proves inequality \eqref{eqn: PPS n/k-apx -2}. 
\end{proof}

\begin{claim}
The approximation factor of Algorithm 1 on input $(f, k)$ is $\Omega(n/k)$.
\end{claim}
\begin{proof}
We note that $f(\{v_0\})=0$ and $f(\{v_i\})=1+a$ for all $i\in[n-1]$.
By Claim \ref{claim: PPS n/k-apx}, on input $(f, k)$, Algorithm 1 returns a partition $\mcp$ consisting of the $k-1$ singleton parts that minimizes $f$ among all singleton sets and the complement of the union of these $k-1$ singleton parts. Therefore, the returned partition $\mcp$ contains $\{v_0\}$ as a part and thus
\[f(\mcp)= \sum_{P\in\mcp:P\neq\{v_0\}}f(P) \geq f(V-\{v_0\})=a(n-1).\]
The second inequality follows from submodularity of the function $f$.
Consider the $k$-partition \linebreak $\{\{v_1\},\{v_2\},\ldots, \{v_{k-1}\},V-\{v_1,\ldots, v_{k-1}\}\}$, which has objective $(1+a)(k-1)+1$. This implies that the optimum $k$-partition $\mcp^\ast$ satisfies $f(\mcp^\ast)\leq (1+a)(k-1)+1$. Thus, the approximation factor of the solution returned by Algorithm 1 is
\begin{align*}
    \frac{f(\mcp)}{f(\mcp^\ast)}\geq \frac{a(n-1)}{(1+a)(k-1)+1}\to\frac{n-1}{k-1}\text{ as }a\to\infty.
\end{align*}
\end{proof}

\section{Conclusion}
\label{sec:conclusion}
The principal partition sequence of submodular functions was shown to  exist by Narayanan \cite{Na91}. 
The principal partition sequence of submodular functions is known in the literature as \emph{principal lattice of partitions} of submodular functions since there exists a lattice structure associated with the sequence of partitions \cite{Cun85, Narayanan-book, DNP03, PN03, Bar00, Kol10, NKI10}. We chose to call it as principal partition sequence in this work since the sequence suffices for our purpose. 
Narayanan, Roy, and Patkar \cite{NRP96} used the principal partition sequence to design an algorithm for submodular $k$-partition. They analyzed the approximation factor of their algorithm for certain subfamilies of submodular functions that arise from hypergraphs.  In this work, we investigated the approximation factor of their algorithm for three broad subfamilies of submodular functions---namely monotone, symmetric, and posimodular submodular functions. Our results show that the principal partition sequence based algorithm achieves the best possible asymptotic approximation factor for all these three subfamilies. A novelty of our contributions is the improvement in the approximability of monotone submodular $k$-partition from $2$ to $4/3$, thus matching the inapproximability threshold. It would be interesting to pin down the approximability of special cases of monotone submodular $k$-partition---e.g.,  matroid $k$-partition and coverage $k$-partition
which are interesting by themselves since they capture several interesting partitioning problems.

\paragraph{Acknowledgements.} Karthekeyan Chandrasekaran would like to thank Chandra Chekuri for asking about the approximation factor of the principal partition sequence based approach for symmetric submodular $k$-partition. 

\bibliographystyle{amsplain}
\bibliography{references}

\end{document}